\documentclass[a4paper, 10pt]{article}
\usepackage[T1]{fontenc}
\usepackage{amssymb,amsfonts,amsmath,amsthm,euscript,makeidx,pifont}
\usepackage[colorlinks=true]{hyperref}
\usepackage{mathrsfs}
\usepackage{graphics}
\usepackage{tikz}

\newtheorem{theorem}{Theorem}[section]

\newtheorem{remark}{Remark}[section]
\newtheorem{proposition}{Proposition}[section]
\newtheorem{lemma}{Lemma}[section]
\newtheorem{definition}{Definition}[section]

\newtheorem{ass}{Assumption}[section]

\newcommand{\er}{\mathbb{R}}
\newcommand{\WW}{\mathbb{W}}
\newcommand{\cW}{{\mathcal{W} }}
\newcommand{\ocW}{{\overline{\mathcal{W} }}}

\newcommand{\bi}{{\bar{i}}}
\newcommand{\Ca}{\mathscr{A}}

\newcommand{\SCM}{\mathscr{C}}
\newcommand{\cS}{\mathcal S}
\newcommand{\cI}{\mathcal I}

\newcommand{\pguess}{{\tau^{\text{guess}}}}
\newcommand{\opguess}{\overline{\tau}^{\text{guess}}}

\newcommand{\dom}{\text{Dom}}
\newcommand{\Of}{{\mathcal O}} 
\newcommand{\Off}{{\mathcal O} \!\!\!\!{\mathcal O} }
\newcommand{\bfs}{{\bf s}}
\newcommand{\bfA}{{\bf A}}

\newcommand{\bfC}{{\bf C}}

\newcommand{\op}{\overline{p}}
\newcommand{\uup}{\underline{p}}
\newcommand{\plol}{p_{\text{lolc}}}

\newcommand{\pcar}{{p^{\text{CO}_{2}}}}
\newcommand{\Pnlt}{{\mathfrak{p}}}
\newcommand{\SSigma}{{\Large \bf \Sigma}}

\newcommand{\sa}{s^{\ast}}
\newcommand{\bfsa}{{\bf s}^{\ast}}
\newcommand\ind{\leavevmode\hbox{\rm \small1\kern-0.35em\normalsize1}}
\newcommand{\pelec}{{p^{\text{elec}}}}
\def\co2{$\text{CO}_{2}$}
\def\etc{\emph{etc\/}}

\usepackage{color}
\definecolor{Violet}{RGB}{200,0,200}

\title{Nash equilibrium for coupling of \co2 allowances and electricity markets\thanks{This work was partly supported by Grant  0805C0098 from ADEME.}}
\author{Mireille Bossy\thanks{Inria, mireille.bossy@inria.fr} \and 
Nadia Ma\"{i}zi\thanks{MinesParisTech, nadia.maizi@mines-paristech.fr} \and 
Odile Pourtallier\thanks{Inria, odile.pourtallier@inria.fr}}
\date{\today}

\begin{document}
\maketitle
\begin{abstract}
In this note, we present an existence result of a Nash equilibrium between electricity producers selling their production on an electricity market and buying \co2 emission allowances on an auction carbon market. The producers' strategies integrate the coupling of the two markets via the cost functions of the electricity production. We set out a clear Nash equilibrium that can be used to compute equilibrium prices on both markets as well as the related electricity produced and \co2 emissions covered.
\end{abstract}

\section{Introduction}  
The aim of this paper is to develop analytic tools, in order to design a relevant mechanism for carbon markets, where relevant refers to emission reduction. For this purpose, we focus on electricity producers in a power market linked to a carbon market. In this context, where the number of agents is limited, a standard game theory approach applies.
The producers are considered as players behaving on the two financial markets represented here by carbon and electricity.  We establish a Nash equilibrium for this non-cooperative $J$-player game through a coupling mechanism between the two markets.
  
The original idea comes from the French electricity sector, where the spot electricity market is often used to satisfy peak demand. Producers behavior is demand driven and linked to the maximum level of electricity production. Each producer strives to maximize its market share.  In the meantime, it has to manage the environmental burden associated with its electricity production through a mechanism inspired by the  EU ETS\footnote{European Emission Trading System}  framework~: each producer emission level must be balanced by a permit or through the payment of a penalty. Emission permit allocations are simulated through a carbon market that allows the producers to buy the allowances at an auction. Our focus on the electricity sector is motivated by its introduction in phase III of the EU ETS, and its prevalence in the emission share. In the present paper, the design assumptions made on the carbon market are dedicated to foster emissions reduction in the entire electricity sector.

Based on a static elastic demand curve (referring to the times stages in an organized electricity market, mainly day-ahead and intra-day), we solve the local problem of establishing a non-cooperative Nash equilibrium for the two coupled markets. 

While literature mainly addresses profit maximization, our share maximization approach deals with profit through specific assumptions: sale at no loss and striking a balance between the purchase of allowances and the carbon footprint of the electricity generated.
Here the market is driven through demand dynamics rather than electricity spot prices dynamics as it has been done in recent works (see \cite{carmona-coulon-schwarz-13a}\cite{carmona-coulon-schwarz-13} \cite{carmona-delarue-etal-13}). 

In Section \ref{sec:market-rules}, we formalize the markets (carbon and electricity) rules and the associated admissible set of players' coupled strategies.
We then first study the Nash equilibrium  on the electricity market alone (see Proposition \ref{propo-Nash}).
Section \ref{sec:nash} is devoted to our Nash equilibrium results.

\section{Coupling markets mechanism}\label{sec:market-rules}
\subsection{Electricity market}
In the electricity market, the demand is  aggregated and summarized by a function  $p\mapsto D(p)$, where $D(p)$ is the quantity of electricity that buyers are ready to obtain at maximal unit price  $p$.  We assume the following~:
\begin{ass}\label{hypo:demande}
The demand function $D(\cdot):[0,+\infty)\rightarrow[0+\infty)$  is decreasing, left continuous, and such that  $D(0) >0$.
\end{ass}

Each producer $j \in \{1, \ldots, J \}$ is characterized by a finite production capacity $\kappa_j$ and a 
bounded and increasing function 
$ c_{j}: [0,\kappa_{j}] \longrightarrow \er^{+}$ that associates a  marginal production cost to any quantity $q$ of 
electricity. These marginal production costs depend on several exogenous parameters 
reflecting the technical costs associated with electricity production e.g. energy prices, O\&M costs, taxes, carbon penalties \etc... This parameter dependency makes possible to build different market coupling mechanisms. In the following we use it to link the carbon and the  electricity markets. 

The merit order ranking features marginal cost functions sorted according to their production costs. These are therefore increasing staircase functions whereby each stair refers to the marginal production cost of a specific unit owned by the producer.

The producers trade their electricity on a dedicated market. For a given producer $j$, the strategy  consists in a function that makes it possible to establish an ask price on the electricity market, defined as
\begin{align*}
s_{j} : & \SCM_j \times \er^{+} \longrightarrow \er^{+} \\
& (c_{j}(\cdot), q) \longrightarrow s_{j}(c_{j}(\cdot), q), 
\end{align*}
where $\SCM_j$ the set of  marginal production cost functions  {\bf are explicitly given in the following} (see \eqref{def:set-C_j}). { \\}
$s_{j}(c_{j}(\cdot), q)$ is the unit price at which the producer is ready to sell quantity  $q$ of electricity. An admissible strategy  fulfills the following sell at no loss constraint
\begin{equation}\label{contrainteStrategie}
s_{j}(c_{j}(\cdot), q) \geq c_{j}(q), \quad \forall q \in  \dom(c_{j}).
\end{equation}
For example we can take $s_{j}(c_{j}(\cdot), q) = c_{j}(q)$ or $s_{j}(c_{j}(\cdot), q) = c_{j}(q)+ \lambda(q)$, where $\lambda(q)$ stands for any additional  profit. \\
As mentioned in the introduction, the  constraint \eqref{contrainteStrategie} guarantees profitable trade in as much as equilibrium established through this class of strategy will bring benefit to each producer. This establishes a link between market share maximization and profit maximization paradigms.

Let us denote $\cS$ as the  class of admissible strategy profiles  on electricity market. We have
\begin{equation}\label{classeStratAdmiss}
\begin{aligned}
\cS = \left \{ 
\begin{array}{l}
\begin{array}{rcl}
\bfs = (s_1,\ldots,s_j); \;  s_{j}: \SCM_j \times \er^{+} & 
\longrightarrow & 
\er^{+} \\ 
(c_{j}(\cdot), q) &
\longrightarrow &
s_{j}(c_{j}(\cdot), q)
\end{array} \\ 
\begin{array}{l}
\mbox{ such that }s_{j}(c_{j}(\cdot), q) \geq c_{j}(q), \quad \forall q \in \dom(c_{j})
\end{array}
\end{array}
\right\}.
\end{aligned}
\end{equation}
As a function of $q$, $s_{j}(c_{j}(\cdot),q)$ is bounded on $\dom(c_{j})$. 
For the sake of clarity, we define for each $q \not \in \dom(c_{j}) $, $s_{j}(c_{j}(\cdot),q) = \plol$, where $\plol$ is the loss of load cost, chosen as any overestimation of the  maximal production costs.

For producer $j$'s  strategy  $s_{j}$,  we define the associated ask size at price   $p$ as 
\begin{equation}\label{defOffrej}
\Of(c_{j}(\cdot),s_{j};p) := \sup\{q, \; s_{j} (c_{j}(\cdot), q) < p \}.
\end{equation}
Hence $\Of(c_{j}(\cdot),s_{j};p)$ is the maximum quantity of electricity at unit price $p$  supplied by producer $j$ on the market. 

\begin{remark}\label{property:offre croissante}
\item[(i)] The ask size function  $p\mapsto \Of(c_{j}(\cdot),s_{j};p)$  is, with respect to $p$,  an increasing surjection from   $[0,+\infty)$ to $[0,\kappa_j]$, right
continuous and such that $\Of(c_{j}(\cdot),s_{j};0)=0$.  
For an increasing  strategy  $s_{j}$, $\Of(s_{j};.)$ is its generalized inverse function with respect to $q$.
\item[(ii)] Given two strategies  $q\mapsto s_{j}(c_{j}(\cdot), q)$ and  $q\mapsto s_j'(c_{j}(\cdot), q)$  such that $s_{j}(c_{j}(\cdot), q) \leq 
s_j'(c_{j}(\cdot), q)$,  for all $q\in \dom(c_{j})$ 
we have for any positive $p$
\begin{equation*}
\Of(c_{j}(\cdot),s_{j};p) \geq \Of(c_{j}(\cdot),s_j';p).
\end{equation*}
Indeed, if $p_{1} \geq p_{2}$ then $\{ q, \; s_{j}(c_{j}(\cdot), q) \leq p_{2} \} \subset \{ q, \; s_{j}(c_{j}(\cdot), q) \leq p_{1} \}$  from which we deduce that $\Of(c_{j}(\cdot),s_j;\cdot)$ is
increasing. 
Next, if $s_j(c_{j}(\cdot),\cdot) \leq s_j'(c_{j}(\cdot),\cdot)$, for any fixed $p$, we have  
$\{ q, \; s_{j}'(c_{j}(\cdot),q) \leq p \} \subset \{ q, \; s_{j}(c_{j}(\cdot), q) \leq p \}$ from which 
the reverse order follows for the asks. 
\end{remark}

We now describe the electricity market clearing. Note that from the market view point, the dependency of the offers with respect to the marginal cost does not need to be explicit. For the sake of clarity, 
we will write $s_{j}(q)$ and $\Of(s_j;p)$ instead of $s_{j}(c_{j}(\cdot),q)$, $\Of(c_{j}(\cdot),s_j;p)$. The dependency will be expressed explicitly whenever needed.

By aggregating the  $J$  ask size functions, we can define the overall supply function $p\mapsto \Off(\bfs;p)$ for producers strategy 
profile $\bfs= (s_{1}, \ldots, s_{J})$ as~:
\begin{equation}
\Off(\bfs; p) = \sum_{j = 1}^{J} \Of(s_{j};p).
\end{equation}
Hence, for any producer strategy profile  $\bfs$, $\Off(\bfs ; p)$ is the quantity of electricity that can be sold on 
the market at unit price  $p$.  

The overall supply function $p\mapsto \Off(\bfs; p)$ is an increasing surjection defined from  $[0,+\infty)$ to $[0,\sum_{j=1}^J\kappa_j]$,
such that $\Off(\bfs;0)=0$. 

\subsubsection{Electricity market clearing}
Taking producers strategy profile $\bfs= (s_{1}(\cdot), \ldots, s_{J}(\cdot))$ the market sets the electricity market
price $\pelec(\bfs)$ together with the quantities $(q_{1}(\bfs), \ldots, q_J(\bfs))$ of electricity sold by each producer. 

The market clearing price $\pelec(\bfs)$ is the unit price paid to each producer for the quantities $q_{j}(\bfs)$ of 
electricity. The price  $p(\bfs)$ may be defined as a price whereby offer satisfies the demand. As  we are working with a general non-increasing demand curve (eventually locally inelastic), the price that satisfies the demand is not necessarily unique.  We thus define the clearing price generically with the following definition.  
\begin{definition}[The clearing electricity price.]\label{def:clearingElec}
Let us define
\begin{equation}\label{reacmarche-prix}
\begin{aligned}
&\uup(\bfs) = \inf \left\{ p > 0 ; \;  \Off(\bfs ; p) >  D(p) \right\} \\
\mbox{ and }\quad&\\
&\op(\bfs) = \sup \left\{ p\in [\uup(\bfs),\plol];  D(p) = D(\uup(\bfs))\right\}
\end{aligned}
\end{equation}
with the convention that $\inf\emptyset = \plol$. The clearing price may then be established as any  $\pelec(\bfs) \in [\uup(\bfs), \op(\bfs)]$ as an output of a specific market clearing rule. For consistency of the price, the market rule  must be such that for any two strategy profiles $\bfs$ and $\bfs '$, 
\begin{equation}\label{regleChoixPrix}
\begin{aligned}
\mbox{if } \uup({\bfs}) < \uup({\bfs '}) \mbox{ then } \pelec({\bfs}) <  \pelec({\bfs '}), \\
\mbox{if } \uup({\bfs}) = \uup({\bfs '}) \mbox{ then } \pelec({\bfs}) =  \pelec({\bfs '}). 
\end{aligned}   
\end{equation}
\end{definition}

Note that $\uup(\bfs)\neq \op(\bfs)$ only if the demand curve $p\mapsto D(p)$ is constant on some interval $[\uup(\bfs),\uup(\bfs)+ \epsilon]$. 

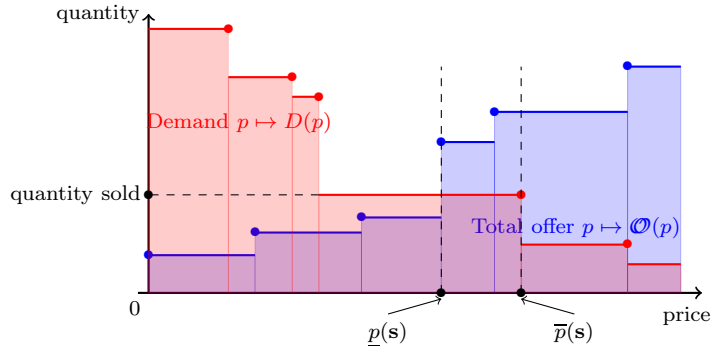
\begin{figure}[ht]
\begin{center}
\begin{tikzpicture}[xscale=7,yscale=0.02]\footnotesize
 \newcommand{\xone}{-.02}
 \newcommand{\xtwo}{ 1.04}
 \newcommand{\yone}{-.4}
 \newcommand{\ytwo}{185}

\begin{scope}<+->;

\draw[black] (0,0) node[anchor=north east] {$0$};

\draw[black,thick,->] (\xone, 0) -- (\xtwo, 0);
\draw[black,thick,] (0.95, -15)  node[right] {price};
\draw[black,thick,->] (0, \yone) -- (0, \ytwo)node[left] {quantity};
\end{scope}
\begin{scope}[thick,blue]
\draw (0,25) node {$\bullet$} ;
\draw (0,25) -- (0.2,25);
\filldraw[very thin,opacity=.2] (0,0.0) rectangle (0.2,25);

\draw (0.2,40) node {$\bullet$} ;
\draw (0.2,40) -- (0.4,40);
\filldraw[very thin,opacity=.2] (0.2,0) rectangle (0.4,40);

\draw (0.4,50) node {$\bullet$} ;
\draw (0.4,50) -- (0.55,50);
\filldraw[very thin,opacity=.2] (0.4,0) rectangle (0.55,50);

\draw (0.55,100) node {$\bullet$} ;
\draw (0.55,100) -- (0.65,100);
\filldraw[very thin,opacity=.2] (0.55,0) rectangle (0.65,100);
\draw (0.59,43) node[right] {Total offer $p\mapsto \Off(p)$};

\draw (0.65,120) node {$\bullet$} ;
\draw (0.65,120) -- (0.9,120);
\filldraw[very thin,opacity=.2] (0.65,0) rectangle (0.9,120);

\draw (0.9,150) node {$\bullet$} ;
\draw (0.9,150) -- (1.0,150);
\filldraw[very thin,opacity=.2] (0.9,0) rectangle (1.0,150);
\end{scope}

\begin{scope}[thick,red]
\draw (0.17,100) node[right,above] {Demand $p \mapsto D(p)$} ;
\draw (0.15,175) node {$\bullet$} ;
\draw (0,175) -- (0.15,175);
\filldraw[very thin,opacity=.2] (0.0,0.0) rectangle (0.15,175);

\draw (0.27,143) node {$\bullet$} ;
\draw (0.15,143) -- (0.27,143);
\filldraw[very thin,opacity=.2] (0.15,0) rectangle (0.27,143);

\draw (0.32,130) node {$\bullet$} ;
\draw (0.27,130) -- (0.32,130);
\filldraw[very thin,opacity=.2] (0.27,0) rectangle (0.32,130);

\draw (0.7,65) node {$\bullet$} ;
\draw (0.32,65) -- (0.7,65);
\filldraw[very thin,opacity=.2] (0.32,0) rectangle (0.7,65);

\draw (0.9,32) node {$\bullet$} ;
\draw (0.7,32) -- (0.9,32);
\filldraw[very thin,opacity=.2] (0.7,0) rectangle (0.9,32);

\draw (0.9,19) -- (1,19);
\filldraw[very thin,opacity=.2] (0.9,0) rectangle (1,19);
\end{scope}

\begin{scope}[black]
\draw[dashed] (0.55,0.0) -- (0.55,150);
\draw (0.55,0.0) node {$\bullet$} ; 
\draw[thin,<-]  (0.55,-4) -- (0.45,-14) node[below] {$\uup(\bfs)$};

\draw[dashed] (0.7,0.0) -- (0.7,150);
\draw (0.7,0.0) node {$\bullet$} ; 
\draw[thin,<-]  (0.7,-4) -- (0.80,-14) node[below] {$\op(\bfs)$};

\draw[dashed]  (0.32,65)  --  (-0.0,65)  node {$\bullet$} node[left]{quantity sold};
\end{scope}
\end{tikzpicture}
\end{center}
\label{clearing} 
\end{figure}

Note also that price $\uup(\bfs)$ is well defined in the case 
where demand does not strictly decrease. This includes the case where demand is constant. 
In such case, $\uup(\bfs)=\plol$ only if the demand curve never crosses the offer. 

Next, we define the quantity of electricity sold at price $\pelec(\bfs)$. 
When $\Off(\bfs;\pelec(\bfs)) \leq D(\pelec(\bfs))$, each producer sells $\Of(\bfs_{j};\pelec(\bfs))$, but cases where  $\Off(\bfs;\pelec(\bfs)) >  D(\pelec(\bfs))$ may occur, requiring the introduction of an auxiliary rule to share $D(\pelec(\bfs))$ among the producers that propose $\Off(\bfs;\pelec(\bfs))$. In such a case, $\uup(\bfs)$ is a discontinuity point of 
$\Off(\bfs;\cdot)$ and/or $\uup(\bfs) < \pelec(\bfs)$. We can split the offer as follows: 
\begin{align*}
\Off(\bfs ; \pelec(\bfs)) = \sum_{j=1}^J \Of(s_j;\uup(\bfs)^-) + \sum_{j=1}^J \Delta^- \Of(s_j;\pelec(\bfs)), 
\end{align*}
where
$ \Delta^- \Of(s_j;\pelec(\bfs)) := \Of(s_{j};\pelec(\bfs)) - \Of(s_{j}; \uup(\bfs)^{-})$.

The market's choice is to fully accept the ask size of producers with continuous ask size curve at point $\uup(\bfs^-)$. For producers with discontinuous ask size curve at $\pelec(\bfs)$, a market rule based on proportionality that favors abundance, is used to share the remaining part of the supply.  More precisely we define $\varphi_{j}(\bfs)$, the quantity of electricity sold by $j$, as 
\begin{equation}\label{reacmarche-qantite-elec}
\begin{aligned}
\quad \varphi_{j}(\bfs) = \left\{
\begin{array}{l}
\Of(\bfs_{j}; \pelec(\bfs)),\\
\quad \mbox{ if }D(\pelec(\bfs))\geq \Off(\bfs ; \pelec(\bfs)),  \\ \\ 
\Of(\bfs_{j}; \uup(\bfs)^-)+\Delta^-\Of(\bfs_{j};\pelec(\bfs))\dfrac{D(\pelec(\bfs)) - \Off(\bfs ; \uup(\bfs)^{-})}{\displaystyle \Delta^- \Off(\bfs ; \pelec(\bfs))},\\
\quad \mbox{ if }D(\pelec(\bfs)) < \Off(\bfs ; \pelec(\bfs)),
\end{array}\right.
\end{aligned}
\end{equation}
where
$
 \Delta^- \Off(\bfs ; \pelec(\bfs)) := \sum_{j=1}^{J}  \Delta^- \Of(s_j;\pelec(\bfs)) > 0$. \\
Note that, when $D(\pelec(\bfs)) < \Off(\bfs ; \pelec(\bfs))$ then $\Delta^- \Off(\bfs ; \pelec(\bfs)) > 0$.  
Note also that we always have  
\begin{align}
\sum_{i=1}^J \varphi_{j}(\bfs)  = D(\pelec(\bfs))\wedge \Off(\bfs ; \pelec(\bfs)).
\end{align}

\subsection{Carbon market}

Producers are penalized according to their emission level if they do not own allowances. Hence
independently from their position on the  electricity market,  producers  buy \co2 emission allowances on a \co2 auction market. 
This market has a finite known quantity $\WW$ of \co2  emission allowances available.

On this market, producers adopt a strategy that consists in a series of bids which may be reorganized in a decreasing  function $w \mapsto A_{j}(w)$ 
defined from  $[0,+\infty)$ to $[0,+\infty)$.  Quantity $A_{j}(w)$ is the unit price that producer $j$ is ready to pay for
quantity $w$ of  \co2 allowance. $\Ca$ denotes the strategy profile set on the \co2 market,
\[
\Ca := \{ \bfA = (A_{1},\ldots ,A_{J}); \mbox{s.t. }A_k:[0,+\infty)\rightarrow[0,+\infty) \mbox{ is decreasing } \}.
\]
Strategy  $ A_{j}$ is associated with a supply (to buy) function, denoted by $p \mapsto \Theta(A_j;p)$. The quantity $\Theta(A_j;p)$ is
the maximum quantity that producer $j$ is ready to buy at price $p$. It is a decreasing left continuous function defined
as 
\begin{align*}
\Theta (A_j;p) = \sup\{w,\;A_j(w)\geq p\}.
\end{align*}
The \co2 market reacts by aggregating the $J$ offers by
$
{{\bf{\Theta}}}(\bfA; p) = \sum_{j=1}^J \Theta (A_j;p),
$
and the clearing market price is established following a {\it second item auction} as~:
\begin{align}\label{reacmarche-prix-quotas}
\pcar(\bfA) := \inf\{p,\;{\bf{\Theta}}(\bfA; p)< \WW\}.
\end{align}

\begin{figure}
\begin{center}
\begin{tikzpicture}[xscale=8,yscale=0.02]\footnotesize
 \newcommand{\xone}{-.02}
 \newcommand{\xtwo}{ 1.04}
 \newcommand{\yone}{-.4}
 \newcommand{\ytwo}{185}

\begin{scope}<+->;
\draw[black] (0,0) node[anchor=north east] {$0$};
\draw[black,thick,->] (\xone, 0) -- (\xtwo, 0);
\draw[black,thick,] (0.95, -15)  node[right] {$w$};
\draw[black,thick,->] (0, \yone) -- (0, \ytwo)node[left] {price};
\end{scope}

\begin{scope}[thick,Violet]
\draw (0.20,100) node[right,above] {Aggregated bid curve} ;
\draw (0.15,175) node {$\bullet$} ;
\draw (0,175) -- (0.15,175);
\filldraw[very thin,opacity=.2] (0.0,0.0) rectangle (0.15,175);

\draw (0.27,143) node {$\bullet$} ;
\draw (0.15,143) -- (0.27,143);
\filldraw[very thin,opacity=.2] (0.15,0) rectangle (0.27,143);

\draw (0.40,130) node {$\bullet$} ;
\draw (0.27,130) -- (0.40,130);
\filldraw[very thin,opacity=.2] (0.27,0) rectangle (0.40,130);

\draw (0.7,65) node {$\bullet$} ;
\draw (0.40,65) -- (0.7,65);
\filldraw[very thin,opacity=.2] (0.40,0) rectangle (0.7,65);

\draw (0.9,32) node {$\bullet$} ;
\draw (0.7,32) -- (0.9,32);
\filldraw[very thin,opacity=.2] (0.7,0) rectangle (0.9,32);

\draw (0.9,19) -- (1,19);
\filldraw[very thin,opacity=.2] (0.9,0) rectangle (1,19);
\end{scope}

\begin{scope}[black]
\draw[dashed] (0.55,0.0) -- (0.55,150) node[above]{carbon clearing};
\draw (0.55,0.0) node {$\bullet$} ; 
\draw[thin]  (0.55,-4) node[below] {$\WW$};

\draw[dashed]  (0.40,65)  --  (-0.0,65)  node {$\bullet$} node[left]{$\pcar$};
\end{scope}

\end{tikzpicture}
\caption{Clearing on the allowances market}
\end{center}
\label{clercarb}
\end{figure}
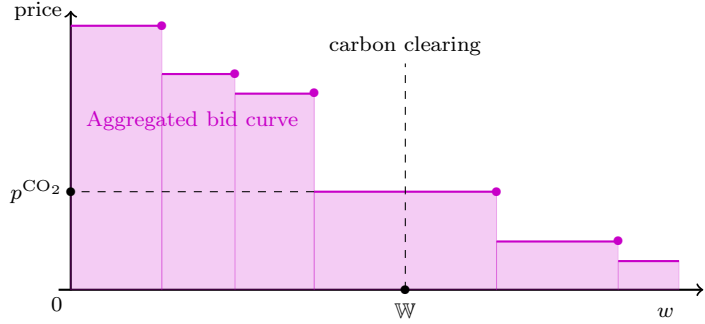
Note that $\pcar(\bfA) =0$ indicates that there are too many allowances. It is worth a reminder here, that the aim of allowances is to decrease emissions. In section \ref{sec:design}, we discuss a design hypothesis (assumption \ref{ass:hypoTW} ) that guarantees an equilibrium price $\pcar(\bfA) >0$.
Therefore, in the following, we assume that the overall quantity $\WW$ of allowances, is such that $\pcar(\bfA) >0$.

By Definition \eqref{reacmarche-prix-quotas}, we have 
${{\bf{\Theta}}}(\bfA ; \pcar(\bfA))\geq \WW ~\mbox{ and }~ {{\bf{\Theta}}}(\bfA ; {\pcar}(\bfA)^{+})\leq \WW.$\\

Producers with  $(\Theta(A_j;\pcar(\bfA))>0$ each obtain the following quantity $\delta_j(\bfA)$ of  allowances 
 \begin{equation}\label{reacmarche-qantite-quotas}
\delta_{j}(\bfA) := \left\{ 
\begin{array}{l}
\Theta(A_{j}; \pcar(\bfA)),\\
\quad \mbox{ if }\Delta^+\Theta(A_j;\pcar(\bfA)) = 0, \\ \\
\Theta(A_{j}; {\pcar(\bfA)^{+}}) + \Delta^+\Theta(A_j;\pcar(\bfA)) \dfrac{\left( \WW -  \Theta(\bfA;{\pcar}(\bfA)^+)\right)}{\Delta^+ \Theta(\bfA;{\pcar}(\bfA))}
,\\
\quad \mbox{ otherwise}
\end{array}\right.
\end{equation}
where $\Delta^+ f(x) := f(x) - f(x^+)$. 

\subsection{Carbon and electricity markets coupling}

As mentioned earlier, for each producer, the marginal cost function is parametrized by the positions 
$\bfA$ of the producers  on the carbon market. Indeed, producer $j$ can obtain  \co2 emission 
allowances on the market to avoid penalization for  (some of) its emissions. Those emissions that are not covered by allowances are penalized at a unit rate $\Pnlt$. 

A profile of an offer to buy from the producers  $\bfA = (A_1, \ldots, A_J)$, through the \co2 market clearing, corresponds to a 
unit price of $\pcar(\bfA)$ of the allowance and quantities $\delta_{j}(\bfA)$ of allowances  bought by each producer (defined by the market rules  
\eqref{reacmarche-prix-quotas},\eqref{reacmarche-qantite-quotas}).

We assume that for all producers the emission rate, $e_{j}$, is constant. Then, the  marginal production cost function $c_{j}^\bfA(\cdot)$, parametrized by the emission regulations, comes out as 
\begin{equation}\label{coutsRegulation}
q\mapsto c_{j}^\bfA(q) = \left \{ 
\begin{array}{ll}
c_j(q) + {e_j} \pcar(\bfA),&  \mbox{ for }  0 < q \leq \displaystyle\frac{\delta_{j}(\bfA)}{e_{j}} \\
c_j(q) + {e_j} \Pnlt, &\mbox{ for }   \displaystyle\frac{ \delta_{j}(\bfA) }{e_{j}} < q \leq \kappa_{j},\\
\end{array}
\right.
\end{equation}
where $c_j(\cdot)$ stands for the marginal production cost without any emission regulation. 

In this coupled market  setting, the strategy of producer  $j$ thus makes a pair $(A_{j}, s_{j})$. The set of admissible strategy profile is defined as 
\begin{align*}
\SSigma = \left \{ (\bfA,\bfs); \;\bfA\in\Ca, \bfs\in \cS \right\}, 
\end{align*}
where in the definition \eqref{classeStratAdmiss}, we use  
\begin{align}\label{def:set-C_j}
\SCM_j = \left\{ c^\bfA_j; \;\bfA \in \Ca \right\}. 
\end{align}
To any strategy profile ${(\bfA,\bfs)} \in \SSigma$, through the market mechanisms described, corresponds  prices for allowances and electricity, $\pcar({(\bfA,\bfs)})$ and $\pelec({(\bfA,\bfs)})$, quantities of allowances bought by each producer, $\delta_{j}({(\bfA,\bfs)})$ and market shares on electricity market $\varphi_{j}({(\bfA,\bfs)})$ of each producer.

\section{Nash Equilibrium}\label{sec:nash}

\subsection{Definition}

We suppose that the $J$ producers behave non cooperatively, aiming at maximizing their individual market share on the electricity market.
For a strategy profile ${(\bfA,\bfs)} \in \SSigma$, the market share of a producer $j$ depends upon its strategy 
$(A_{j},s_{j}(\cdot))$ but also on the strategies $(\bfA_{-j},\bfs_{-j})$ of the other producers
\footnote{Here ${\bf v}_{-j}$ stands for
the profile $(v_{i},\cdots, v_{j-1},v_{j+1},\cdots, v_{J})$.}. 
In this set-up the natural solution is the Nash equilibrium (see e.g. \cite{basar-olsder-98}).
More precisely we are looking for a strategy profile 
\[
{(\bfA^{\ast},\bfs^{\ast})} = ( (A_{1}^{\ast},s_{1}^{\ast}), \cdots, (A_{J}^{\ast},s_{J}^{\ast}) ) \in \SSigma
\]
that 
satisfies Nash equilibrium conditions:  none of the 
producers would strictly benefit, that is, would strictly increase its market share from a unilateral deviation. Namely, for any producer $j$ strategy ${(\bfA_{j},\bfs_{j})}$ such that 
$({(\bfA^{\ast}_{-j},\bfs^{\ast}_{-j})}; {(A_j,s_j)}) \in \SSigma $, we have
\footnote{$({\bf v}_{-j} ; v)$ stands for $(v_{1},\cdots v_{j-1}, v, v_{j+1}, \cdots v_{J})$}

\begin{align}\label{NashGlob}
\varphi_{j}({{(\bfA^{\ast},\bfs^{\ast})}}) \geq  \varphi_{j}({{(\bfA^{\ast}_{-j},\bfs^{\ast}_{-j})}};{(A_j,s_j)}), 
\end{align}
where $q_j$ is the quantity of electricity sold. Note that the dependency in terms of $\bfA$ through the marginal cost $c_j^\bfA$ is now explicit in $q_j$. 

Condition \eqref{NashGlob} has to be satisfied for any unilateral deviation of any producer $j$. In particular  \eqref{NashGlob}  has to be satisfied for a producer $j$ 
admissible deviation $(A_{j}^{\ast}, s_{j})$ such that  $({(\bfA^{\ast}_{-j},\bfs^{\ast}_{-j})}; {(A_j^{\ast},s_j)}) \in \SSigma $ of producer $j$ that would change its behavior only on the electricity market. 
Consequently, Nash equilibrium for the electricity  component component $\bfs ^{\ast}$ of the Nash equilibrium is also a Nash equilibrium for 
a game where producers only behave on an electricity market with
marginal production costs $c_{j}^{\bfA^{\ast}}(\cdot)$, $j=1, \cdots J$.

The Nash equilibrium for a game restricted to the electricity
market characterizes the $\bfs^{\ast}$ component for the 
coupled market game equilibrium.

Note that, if $\bfs^{\ast}$ is the  producers behavior on the electricity 
market at the Nash equilibrium, any behavior $\bfA$  on the \co2 market, 
such that the strategy profile $(\bfA,\bfs^{\ast})$ is admissible yields to the 
same market share for each producer. 

Next section focuses on determining a Nash equilibrium on the game restricted to the electricity market.

\subsection{Equilibrium on Power market}

In this restricted set-up, we consider that the marginal costs   $\{c_j,j=1\ldots,J\}$ are known data, possibly fixed through the position $\bfA$ on the \co2 market. In this section, we refer to $\cS$ as the set  of admissible strategy profiles,in the particular case where $\SCM_j=\{c_j\}$ for each $j=1,\ldots,J$. 

The Nash equilibrium problem is as follows : find a strategy profile $\bfsa = (\sa_{1}, \ldots, \sa_{J}) \in \cS$  such that 
\begin{align}\label{nashQuantiteElec}
\begin{aligned}
\forall j, \forall \;s_{j}\neq \sa_{j}, \quad \varphi_{j}(\bfsa) \geq \varphi_{j}(\bfsa_{-j}; s_{j}).
\end{aligned}
\end{align}

The following proposition exhibits a Nash equilibrium, whereby each producer must choose the strategy denoted by $C_{j}$, and referred to as {\it marginal production cost strategy}. It is defined by 
\begin{align}\label{stratCoutMarg}
C_{j}(q) = \left \{ 
 \begin{array}{l}
   c_{j}(q), \mbox{ for } q \in \dom(c_{j}) \\
   \plol , \mbox{ for } q \not \in \dom(c_{j}) .
 \end{array} \right.
\end{align}
\begin{proposition}\label{propo-Nash}
\item[(i)] 
For any strategy profile ${\bfs} = (s_1,\ldots,s_J)$, no producer  $j\in\{1,\ldots,J\}$ can
be penalized by deviating from strategy  $s_j$ to is marginal production cost strategy $C_j$, namely, 
\begin{equation}
\label{propoPti}
\varphi_j(\bfs)\leq \varphi_({\bf{s}}_{-j};C_j).
\end{equation}
In other words, $C_{j}$ is a dominant strategy for any producer $j$.
\item[(ii)] The strategy profile $\bfC =(C_{1},\dots C_{J})$ is a Nash equilibrium.
\item[(iii)] If the strategy profile $\bfs \in \cS$ is a Nash equilibrium, then we have 
$\pelec(\bfs) = \pelec(\bfC)$ and for any producer $j$, $\varphi_{j}(\bfs) = \varphi_{j}(\bfC)$.
\end{proposition}

Point (ii) of the previous proposition is a direct consequence of the dominance property (i).
The proof of both (i) and (iii) can be found in \cite{preprint-BMP}. Point (ii) of the proposition
exhibits a Nash equilibrium strategy profile. Clearly this equilibrium is not  unique since we can easily 
show  that a producer's given supply can follow from countless  different strategies. 
Nevertheless point (iii) shows that for any Nash equilibrium, the associated electricity prices are the same and the 
quantity of electricity bought by any producer $j$ are the same for
all equilibrium profiles.

\subsection{Coupled markets design through Nash equilibrium}\label{sec:design}

From this point we restrict our attention to a particular design of the market. In the following, the scope of the analysis applies to a special class of producers, a specific electricity market price clearing (satisfying Definition \ref{def:clearingElec}), a range of quantities of allowances available on the \co2 market. Although not necessary, the following  restriction brings simplifications to the development. 
\begin{ass}{\bf On the producers.}\label{ass:producers}
Each producer $j$ operates a single production unit (with emission rate $e_{j}$), for which  
\begin{itemize}
\item[(i)] The marginal production cost is as in Equation (\ref{coutsRegulation}), where the
contribution that does not depend on the producer positions $\bfA$ in the \co2 market is constant,
$c_{j}(q) = c_{j}\index{q \in [0,\kappa_{j}]}$.
\item[(ii)] The producers are two by two different~: $\forall i, j \in \{1, \cdots J \}, (c_{i},e_{i}) \neq (c_{j},e_{j})$.
\end{itemize}
\end{ass}
For a given  strategy profile on the electricity market, Definition \ref{def:clearingElec} gives a range of possible determination for the electricity price. 
Previously, the analysis of the  Nash Equilibrium restricted to the  electricity market, did not require a precise clearing price determination. 
Nevertheless to extend our analysis of the coupling we need to explicit this determination and assume the following~:
\begin{ass}{\bf On the market electricity}\label{ass:ElecClearingPrice}
For a given strategy profile $\bfs$ of the producers, the clearing price of electricity is
$p(\bfs) = \op(\bfs)$, where $\op(\bfs)$ is  defined in Definition \ref{def:clearingElec} by Equation (\ref{reacmarche-prix}).
\end{ass}
As previously noted, this choice of electricity price ensures that  for any strategy profile $\bfs$, for any positive $\epsilon$ we have 
$D(p(\bfs)+\epsilon) < D(p(\bfs))$, for any strategy profile $\bfs$. This property is necessary for the  main theorem \ref{propo:Nashcoupled} proof. 

The quantity $\WW$ of \co2 allowances available on the market plays a crucial role in the market 
design.
As a matter of fact, if this quantity is too large, its price on market will drop to zero, leaving 
the market incapable of fulfilling its role of decreasing \co2 emissions. Therefore we clearly
need to make an assumption that restricts the number of allowances available. Capping the maximum 
quantity of allowances available requires information about producers willing to obtain
allowances. This is the objective of the following paragraph where we define a {\it willing to buy} function that plays a central part on the construction of the Nash equilibrium.  

\subsection*{Willing to buy functions}
In this paragraph, we aim at guessing a Nash equilibrium candidate. We base our reasoning on the dominant strategy on the electricity market alone (see Proposition\eqref{propo-Nash}). For a while, we consider an exogenous  \co2 cost $\tau$. The producers marginal cost become for any $\tau \in [0,\Pnlt]$, $c_{j}^{\tau}, \; j=1,\cdots J$,  
$c_{j}^{\tau}(q) = c_{j} + \tau e_{j}, \;\; q \in [0,\kappa_{j}]$.  In this framework, the dominant strategy is also parametrized by $\tau$ as $C^{\tau}_j(\cdot)$ defined as in \eqref{stratCoutMarg}. In the same way, we define the clearing electricity price and quantities in terms of $\tau$ only by 
\begin{align*}
\pelec(\tau) &= \pelec (\{C_j^\tau(\cdot),j=1\ldots,J\})\\
\varphi_j(\tau) &= \varphi_j(\{C_j^\tau(\cdot),j=1\ldots,J\}).
\end{align*}
We determine two {\it willing-to-buy-allowances functions}  $\cW(\cdot)$ and $\ocW(\cdot)$, following a Dutch auction mechanism-like as follows : 
\begin{align}\label{WillingQuota}
&\cW (\tau)  = \sum_{j=1}^J e_{j} \varphi_j(\tau) 
\quad\mbox{ and } \quad
\ocW(\tau) = \sum_{j=1}^J e_{j} \kappa_{j}  
\ind_{\{\varphi_j(\tau)>0\}} 
\end{align}

Given the \co2 cost $\tau$, the amount $\cW (\tau)$ represent the allowances needed to cover the global emissions generated by  the players who won electricity market shares on the electricity market.
$\ocW(\tau)$ represent the allowances needed in the case producers whish to cover their overall production capacity $\kappa_{j}$.  Obviously we have 
$\cW (\tau) \leq \ocW(\tau)$. We now can state our last design assumption,

\begin{ass}{\bf On carbon market design.}\label{ass:hypoTW}
The number $\WW$ of the allowances available on the auction \co2 market satisfies
\[\cW(0) > \WW > \ocW(\Pnlt).\]
\end{ass}

\begin{proposition}\label{pelecCroissante} 
As functions of $\tau$,  $\pelec(\cdot)$  and  $\sum_{j} q_{j}(\cdot)$ are respectively increasing and decreasing.
\end{proposition}

This proposition is a consequence of Remark \ref{property:offre croissante}, through cost parameter $\tau$. 

Assumption \ref{ass:hypoTW} allows us to define two prices of particular interest  for the construction of the equilibrium strategy:
\begin{align}\label{eq:prix_ante_carbon} 
& \pguess = \sup \{\tau\in[0,\Pnlt]~\text{s.t.}~\cW(\tau) > \WW \} \; \mbox{ and } \;\;
\opguess = \sup \{\tau\in[0,\Pnlt]~\text{s.t.}~\ocW(\tau) > \WW\}.
\end{align}
\begin{lemma}\label{lem:couplage-continuite}
\item{(i)} We have  $\cW(\pguess) = \WW$. 
\item{(ii)} $\ocW$ is a staircase  function  valued in the finite set  $\{ \sum_{i\in {\mathcal I}  } \kappa_{j} e_{j}; {\cI}\subset \{ 1,\cdots J\}\}$.
\end{lemma}

Let us define $\cI(\tau) := \{j;\;q_j(\tau) \neq 0\}$. 
\begin{proposition}
At $\opguess$, only one of the following two cases may occur~:
\item{\bf Case A.} There exists a unique producer, denoted $\bar{i}$ such that 
$\cI(\opguess) = \cI({\opguess}^+) \cup \{\bar{i}\}$. 
For  $\bar{i}$ we have
$\frac{1}{e_{\bar{i}}}\left(\pelec(\opguess)  -c_{\bar{i}} \right) = \opguess$ 
and 
$
\frac{1}{e_{\bar{i}}}\left(\pelec(\pguess) - c_{\bar{i}} \right) = \pguess.
$

\item{\bf Case B.} There exists two  producers, denoted $i_{l}$ and $i_{r}$ such that

\[
\cI(\opguess) = \cI({\opguess}^+) \cap \cI(\opguess)\cup \{i_l\}\quad \mbox{and} \quad \cI({\opguess}^+) = \cI({\opguess}^+) \cap \cI(\opguess)\cup \{i_r\}.
\]
\end{proposition}
\begin{proof}
This follows directly from the fact that $\ocW(\cdot)$ is a staircase function, and from the fact 
that the producers are two by two different.
\end{proof}

We define a strategy profile on the coupled market in the two cases.
\begin{definition}
We define the strategy profile 
${(\bfA,\bfs)}^{\ast} = ( (s_{1}^{\ast},A_{1}^{\ast}), \cdots (s_{J}^{\ast},A_{J}^{\ast}))$, 
where $s^{\ast}_{j} := C_{j}$  is the {\it marginal production cost strategy} and 
$A_{j}^{\ast}$ is defined as follows (depending whether any of the two cases occurs):
\item{\bf Case A.}
\begin{equation}
\begin{aligned}
&\mbox{For } \bi,&  \quad w\mapsto A^\ast_{\bi}(w) := & 
\displaystyle \left(\opguess + \delta \right) \ind_{\displaystyle\{w < e_{\bi} \varphi_{\bi}(\opguess) - \epsilon\}} \\
& & & + 
\displaystyle \pguess \ind_{\displaystyle\{ e_{\bi} \varphi_{\bi}(\opguess) - \epsilon < w \leq e_{\bi} \kappa_{\bi} \}}\\
&\mbox{For } k \neq \bi,& \quad  
w\mapsto A^\ast_k(w) := & \frac{1}{e_k}\left( \pelec(\opguess) - c_k\right) \ind_{\{w \leq e_k \kappa_{k} \}} 
\end{aligned} 
\end{equation}

\item{\bf Case B.}
\begin{equation*}
\begin{aligned} 
&\mbox{For } i_l,&  w\mapsto A^\ast_{i_{l}}(w) := &
\displaystyle \left(\opguess + \delta\right) 
\ind_{\displaystyle\{0 < w \leq \WW-\ocW({\opguess}^+)-\epsilon \} } \\
& &  &+ \opguess \ind_{\displaystyle\{\WW-\ocW({\opguess}^+) - \epsilon < w \leq e_{i_l}\kappa_{i_l} \}} \\
~\\
&\mbox{For } i_r,& w\mapsto A^\ast_{i_r}(w) := &  \left(\opguess + \delta \right) \ind_{\{w \leq e_{i_r}\kappa_{i_r}\}} \\
~\\
&\mbox{For } k \not \in \{i_l, i_r\},& w\mapsto A^\ast_k(w) := & \frac{1}{e_k}\left(\pelec(\opguess) - c_k\right)\ind_{\{w\leq e_k\kappa_k\}}. 
\end{aligned}
\end{equation*}
\end{definition}

Now we can state our main result~:
\begin{theorem}\label{propo:Nashcoupled}
Under Assumptions \ref{ass:producers}, \ref{ass:ElecClearingPrice} and \ref{ass:hypoTW}, one can identify  $(\epsilon, \delta)$ such that $(\bfA^\ast, \bfs^\ast )$ is a Nash equilibrium. 
For this equilibrium the following applies
\item{(i)} the carbon price is $\pguess$
\item{(ii)} the electricity price is $\pelec(\pguess)$ 
\item{(iii)} for any producer, the quantity of allowance bought is null if the quantity of electricity sold is null,
\end{theorem}
The proof of this proposition relies on the analysis of any possible deviation. 
It can be found in \cite{preprint-BMP}.
Note that the expression of the equilibrium is explicit, which allows to explicit the quantity of emission. This allows a analysis of the impact of such markets on the overall emission.

\section{Conclusion}
Once emitted into the atmosphere, \co2  will remain there for more than a century. Estimating its value is an essential indicator for efficiently defining policy. Therefore, carbon valuation remains a main issue in order to design markets fostering emission reductions. In this paper, we established the links between an electricity market and a carbon auction market through the analysis of electricity producers strategies. They have been proven to lead to a Nash equilibrium enabling the computation of equilibrium prices on both markets. This equilibrium derives, for each producer, in a level electricity produced and \co2  emissions covered. Beyond the analysis of the Nash equilibrium, we envisage the analysis of the electricity production mix, with a particular focus on renewable shares which do not participate to emissions.


\begin{thebibliography}{1}

\bibitem{basar-olsder-98}
Tamer Ba{\c{s}}ar and Geert~Jan Olsder.
\newblock {\em Dynamic Noncooperative Game Theory, 2nd Edition}.
\newblock Society for Industrial and Applied Mathematics, 1998.

\bibitem{bossy-maizi-al-05}
Mireile Bossy, Nadia Ma{\"{\i}}zi, Geert~Jan Olsder, Odile Pourtallier, and
  Etienne Tanr{\'e}.
\newblock Electricity prices in a game theory context.
\newblock In {\em Dynamic games: theory and applications}, volume~10 of {\em
  GERAD 25th Anniv. Ser.}, pages 135--159. Springer, New York, 2005.

\bibitem{preprint-BMP}
Mireile Bossy, Nadia Ma{\"{\i}}zi, and Odile Pourtallier.
\newblock Design analysis of carbon auction market, through electricity market
  coupling.
\newblock preprint online in HAL, 2013.

\bibitem{carmona-coulon-schwarz-13}
Ren{\'e} Carmona, Michael Coulon, and Daniel Schwarz.
\newblock Electricity price modeling and asset valuation: a multi-fuel
  structural approach.
\newblock {\em Mathematics and Financial Economics}, 7(2):167--202, 2013.

\bibitem{carmona-coulon-schwarz-13a}
Ren{\'e} Carmona, Michael Coulon, and Daniel Schwarz.
\newblock The valuation of clean spread options: Linking electricity, emissions
  and fuels.
\newblock To appear in Quantitative Finance.

\bibitem{carmona-delarue-etal-13}
Ren\'e Carmona, Fran{\c{c}}ois Delarue, Gilles-Edouard Espinosa, and Nizar
  Touzi.
\newblock Singular forward-backward stochastic differential equations and
  emissions derivatives.
\newblock {\em Ann. Appl. Probab.}, 23(3):1086--1128, 2013.

\end{thebibliography}
\end{document}